\documentclass[letterpaper,11pt]{article}
\pdfoutput=1

\usepackage[hmargin=1in,vmargin=1in,marginparsep=0.2in,marginparwidth=0.5in]{geometry}
\usepackage{hyperref}
\usepackage[numbers,sort&compress]{natbib}
\usepackage{amsfonts,amsmath,amssymb,amsthm,booktabs,color,enumitem,graphicx,marginnote,wrapfig}
\usepackage{microtype}

\catcode`\"=13
\def"{\bar}

\newcommand{\mydoi}[2]{\href{http://dx.doi.org/#1}{#2}}
\newcommand{\myeprint}[2]{\href{http://arxiv.org/abs/#1}{#2}}
\newcommand{\myurl}[2]{\href{#1}{#2}}

\newcommand{\G}{\mathcal{G}}
\newcommand{\T}{\mathcal{T}}
\DeclareMathOperator{\pred}{pred}
\DeclareMathOperator{\head}{head}
\DeclareMathOperator{\tail}{tail}
\DeclareMathOperator{\prune}{prune}
\DeclareMathOperator{\ext}{ext}
\DeclareMathOperator{\real}{real}

\DeclareMathOperator{\polylog}{polylog}

\definecolor{notecol}{gray}{0.4}
\newcommand{\notestyle}[1]{{\footnotesize\textcolor{notecol}{#1}\par}}
\newcommand{\note}[1]{\marginnote[\raggedright\notestyle{#1}]{\raggedleft\notestyle{#1}}}

\newtheorem{theorem}{Theorem}
\newtheorem{lemma}[theorem]{Lemma}
\newtheorem{corollary}[theorem]{Corollary}

\theoremstyle{remark}
\newtheorem{remark}{Remark}

\setenumerate{label=(\alph*),leftmargin=7ex}
\hypersetup{
    colorlinks=true,
    linkcolor=black,
    citecolor=black,
    filecolor=black,
    urlcolor=black,
}

\hypersetup{
    pdftitle={Distributed maximal matching: greedy is optimal},
    pdfauthor={Juho Hirvonen, Jukka Suomela},
}

\begin{document}

\vspace*{5ex}
\begin{center}
    {\Large \textbf{Distributed Maximal Matching:\\Greedy is Optimal}\par}

    \bigskip
    \bigskip

    Juho Hirvonen and Jukka Suomela

    \bigskip
    Helsinki Institute for Information Technology HIIT \\
    University of Helsinki
\end{center}

\bigskip
\begin{quote}
    \textbf{Abstract.}
    We study distributed algorithms that find a maximal matching in an anonymous, edge-coloured graph. If the edges are properly coloured with $k$ colours, there is a trivial greedy algorithm that finds a maximal matching in $k-1$ synchronous communication rounds. The present work shows that the greedy algorithm is optimal in the general case: any algorithm that finds a maximal matching in anonymous, $k$-edge-coloured graphs requires $k-1$ rounds.
    
    If we focus on graphs of maximum degree $\Delta$, it is known that a maximal matching can be found in $O(\Delta + \log^* k)$ rounds, and prior work implies a lower bound of $\Omega(\polylog(\Delta) + \log^* k)$ rounds. Our work closes the gap between upper and lower bounds: the complexity is $\Theta(\Delta + \log^* k)$ rounds. To our knowledge, this is the first linear-in-$\Delta$ lower bound for the distributed complexity of a classical graph problem.
\end{quote}

\bigskip
\begin{quote}
    \textbf{Corresponding author:} \\
    Jukka Suomela \\
    Helsinki Institute for Information Technology HIIT \\
    P.O. Box 68, FI-00014 University of Helsinki, Finland \\
    jukka.suomela@cs.helsinki.fi
\end{quote}

\thispagestyle{empty}
\setcounter{page}{0}
\newpage

\section{Introduction}

In the study of deterministic distributed graph algorithms, there are two parameters that are commonly used to describe the computational complexity of a graph problem: $n$, the number of nodes in the graph, and $\Delta$, the maximum degree of the graph. For a wide range of problems, the complexity is well-understood as a function of $n$, but understanding the complexity as a function of $\Delta$ is one of the major open problems in the area. For example, the maximal matching problem can be solved in $O(\Delta + \log^* n)$ rounds \cite{panconesi01some}, while the best lower bound is $\Omega(\polylog(\Delta) + \log^* n)$  \cite{linial92locality, kuhn04what, kuhn06price, kuhn10local}.

The present works gives the first tight lower bound that is linear in $\Delta$ for a classical graph problem. In particular, we study the problem of finding a \emph{maximal matching in anonymous, edge-coloured graphs}. If the edges are $k$-coloured, the problem can be solved in $O(\Delta + \log^* k)$ rounds with an adaptation of a simple deterministic algorithm \cite{panconesi01some}. It is well-known that the complexity is $\Omega(\log^* k)$ rounds \cite{linial92locality}; we close the case by proving a lower bound of $\Omega(\Delta)$ rounds.

\subsection{Related Work}

For many graph problems, the state-of-the-art algorithms are extremely fast even if the network is very large---provided that $\Delta$ is small. For example, the following problems can be solved in $O(\Delta + \log^* n)$ synchronous communication rounds (assuming $O(\log n)$-bit node identifiers):
\begin{itemize}[noitemsep]
    \item maximal matching \cite{panconesi01some},
    \item vertex colouring with $\Delta+1$ colours \cite{barenboim09distributed, kuhn09weak},
    \item edge colouring with $2\Delta-1$ colours \cite{panconesi01some}.
\end{itemize}
There are also problems that can be solved in $O(\Delta)$ rounds, independently of $n$ (even in anonymous networks without unique identifiers):
\begin{itemize}[noitemsep]
    \item maximal matching in $2$-coloured graphs \cite{hanckowiak98distributed},
    \item maximal edge packing \cite{astrand10vc-sc},
    \item $2$-approximation of minimum vertex cover \cite{astrand10vc-sc}.
\end{itemize}
For each of these problems, the dependence on $n$ in the running time is well-understood. In particular, Linial's \cite{linial92locality} lower bound shows that maximal matching, vertex colouring, and edge colouring require $\Omega(\log^* n)$ rounds, even if $\Delta = 2$.

However, we do not yet understand the dependence on $\Delta$. For example, the best known lower bound for the maximal matching problem is \emph{logarithmic} in $\Delta$ \cite{kuhn04what, kuhn06price, kuhn10local}, while the above upper bounds are \emph{linear} in $\Delta$.

Some $\polylog(\Delta)$ upper bounds are known for graph problems. For example, good approximations of fractional matchings can be found in $\polylog(\Delta)$ rounds \cite{kuhn06price}; however, this does not seem to yield a deterministic $\polylog(\Delta)$-time algorithm for any of the above problems. Ha\'n\'ckowiak et al.'s~\cite{hanckowiak01distributed} algorithm finds a maximal matching in $\polylog(n)$ rounds, avoiding the linear dependence on $\Delta$; however, it comes at the cost of a non-optimal dependence on $n$.

It is easy to come up with an artificial problem with the complexity of $\Theta(\Delta)$, but so far no such tight results are known for classical graph problems such as maximal matchings. The lower-bound result by Kuhn and Wattenhofer~\cite{kuhn06complexity} comes close, but it only applies to a restricted family of algorithms.

\subsection{Greedy Maximal Matching}\label{ssec:gmm}

We will focus on the task of finding a maximal matching in an edge-coloured graph, using a deterministic distributed algorithm in a network of anonymous nodes (see Section~\ref{sec:prelim} for formally precise definitions and e.g. the survey~\cite{suomela09survey} for more background information).

If the graph is edge-coloured with $k$ colours, there is a very simple greedy algorithm that solves the problem in $k$ steps: We start with an empty matching $M \gets \emptyset$. Then, in step $i$ we consider all edges of colour $i$ in parallel. If an edge $\{u,v\}$ is of colour $i$, and neither $u$ nor $v$ is matched, we add $\{u,v\}$ to $M$. The following figure illustrates the greedy algorithm for $k = 4$; the thick edges indicate matching~$M$.
\begin{center}
    \includegraphics[page=1]{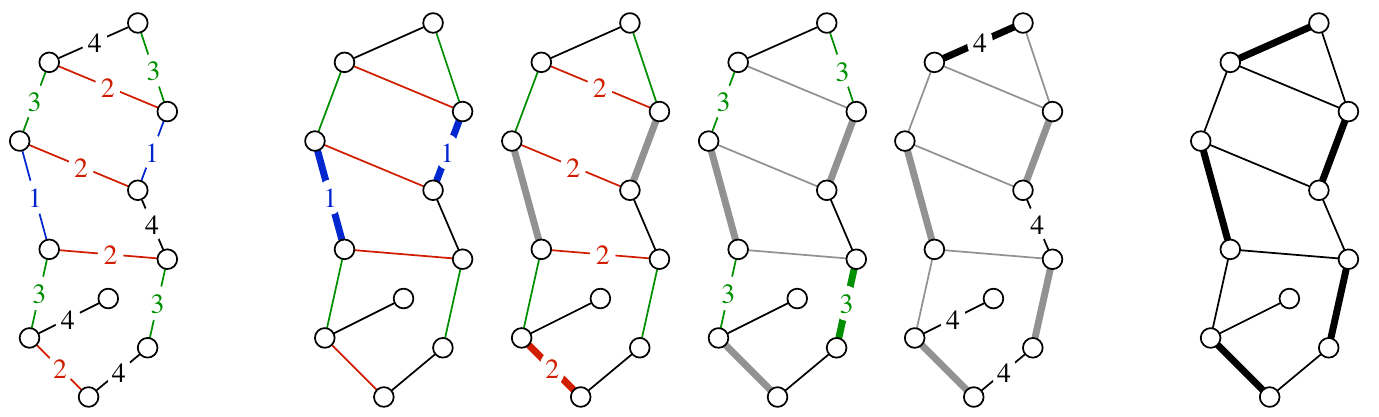} 
\end{center}

To analyse the exact running time of the greedy algorithm, we need to fix the model of computation. As usual, each node is a computational entity and there is an edge if the nodes can exchange messages with each other---the same graph is both the problem instance and the network topology. Throughout this work, the running time is defined to be the number of synchronous communication rounds. Initially, each node knows the colours of its incident edges. In every round, each node in parallel (1)~sends a message to each of its neighbours, (2)~receives a message from each of its neighbours, and (3)~updates its own state. After each round, a node can stop and announce its \emph{local output}: whether it is matched and with which neighbour.

With these definitions, it is straightforward to verify that the running time of the greedy algorithm is at most $k-1$ communication rounds. To see this, note that the first step of the greedy algorithm does not require any communication: if a node has an incident edge of colour $1$, it is matched along this edge. Hence we have the following lemma.
\begin{lemma}\label{lem:ub}
    Let $k$ be a positive integer. There exists a deterministic distributed algorithm with running time $k-1$ that finds a maximal matching in any anonymous, $k$-edge-coloured graph.
\end{lemma}

We can also easily verify that the analysis is tight, i.e., the worst-case running time of the greedy algorithm is exactly $k-1$ rounds. The following figure illustrates a worst-case input for $k = 4$. In the greedy algorithm $u$ is unmatched while $v$ is matched. However, radius-$2$ neighbourhoods of $u$ and $v$ are indistinguishable; in order to produce a different output, we must propagate information over distance $k-1 = 3$: from $x$ to $u$ and from $y$ to $v$. Hence any faithful implementation of the greedy algorithm requires at least $k-1$ communication rounds.
\begin{center}
    \includegraphics[page=2]{figs.pdf} 
\end{center}

Naturally, if our goal is to find a maximal matching, there is a wide range of possible algorithms, and in many special cases we already know how to beat the greedy algorithm. However, we show that \emph{in the general case, the greedy algorithm is optimal}. The main contribution is summarised in the following theorem.

\begin{theorem}\label{thm:lb2}
    Let $k$ be a positive integer. A deterministic distributed algorithm that finds a maximal matching in any anonymous, $k$-edge-coloured graphs requires at least $k-1$ communication rounds.
\end{theorem}

We prove Theorem~\ref{thm:lb2} in Section~\ref{sec:lb}. The lower bound holds even if we allow arbitrarily large messages and unbounded local computations, while the matching upper bound is achieved by a simple algorithm that uses only small messages, little memory, and trivial state transitions.

\subsection{Special Cases}

Let us now return to the case of bounded-degree graph. If $k \gg \Delta$, we can use Cole--Vishkin \cite{cole86deterministic} style colour reduction techniques to considerably speed up the algorithm. For example, a straightforward adaptation of Panconesi and Rizzi's \cite{panconesi01some} algorithm finds a maximal matching in $O(\Delta + \log^* k)$ rounds.

Linial's~\cite{linial92locality} result gives us the lower bound of $\Omega(\log^* k)$; however, so far it has not been known whether $\Omega(\Delta)$ rounds is required. Our result now settles this question. The maximum degree of a $k$-edge-coloured graph is at most $k$, and we have the following corollary.
\begin{corollary}\label{cor:lb}
    A deterministic distributed algorithm that finds a maximal matching in an anonymous edge-coloured graph of maximum degree $\Delta$ requires $\Omega(\Delta)$ communication rounds.
\end{corollary}

Incidentally, our lower-bound construction is a $d$-regular graph with $d = k - 1$, and hence this work shows that we need $d$ rounds even in the seemingly simple case of $d$-regular graphs (assuming $d \ge 2$). Note that in a regular graph, an optimal fractional matching (edge packing) is trivial to find, and none of the existing lower bounds \cite{kuhn04what, kuhn06price, kuhn10local} apply---previously, we have not even had polylogarithmic-in-$\Delta$ lower bounds for such graphs.

Also note that if we study $d$-regular graphs with $d = k$, the problem becomes trivial: the edges of colour $1$ form a perfect matching and we can solve the problem in constant time. The case of $d = k - 1$ is the first non-trivial case, and it is already sufficiently rich to show that the greedy algorithm is optimal.

\subsection{Future Work}

Our lower-bound result covers the case of anonymous networks, including the widely-studied \emph{port-numbering model} \cite{angluin80local, yamashita96computing} and its weaker variants \cite{yamashita99leader} such as the \emph{broadcast model} \cite{astrand10vc-sc}. What remains open is the case of networks in which nodes have unique identifiers. Nevertheless, our result shows that in order to break the $\Omega(\Delta)$ barrier, an algorithm has to make an essential use of the unique node identifiers.

\section{Preliminaries}\label{sec:prelim}

In our lower-bound construction, we will need to manipulate edge-coloured trees, and certain group-theoretic concepts turn out to be useful.

\subsection{Group \texorpdfstring{$G_k$}{Gk}}

Throughout this text, $k$ is a positive integer. We use the shorthand notations \note{$X+x$\\$X-x$\\$[i]$\\$G_k$\\$e$}$X+x = X \cup \{x\}$ and $X-x = X \setminus \{x\}$ for a set $X$, and $[i] = \{1,2,\dotsc,i\}$ for an integer $i$.

We define the group $G_k = \langle 1, 2, \dotsc, k \mid 1^2, 2^2, \dotsc, k^2 \rangle$. That is, the generators of group $G_k$ are $1, 2, \dotsc, k$, and we have the relations $c^2 = cc = e$ for each $c \in [k]$; we use $e$ to denote the identity element, and we use the multiplicative notation $xy$ or $x\cdot y$ for the group operation. Group $G_k$ is the free product of $k$ cyclic groups of order two, a.k.a.\ the group generated by $k$ involutions, the universal Coxeter group, or the free Coxeter group.

Let \note{$\Gamma_k$}$\Gamma_k$ be the Cayley graph of $G_k$ with respect to the generators $[k]$; see Figure~\ref{fig:cayley} for an illustration. In $\Gamma_k$, we have a node for each element $x \in G_k$, and there is an edge of colour $c \in [k]$ from $x \in G_k$ to $y \in G_k$ if $y = xc$. As each generator is its own inverse, there is an edge of colour $c$ from $x$ to $y$ iff there is an edge of colour $c$ from $y$ to $x$; hence we can interpret $\Gamma_k$ as an undirected graph. It can be verified that $\Gamma_k$ is a $k$-regular $k$-edge-coloured tree; $\Gamma_k$ is countably infinite if $k \ge 2$.

\begin{figure}
    \centering
    \includegraphics[page=3]{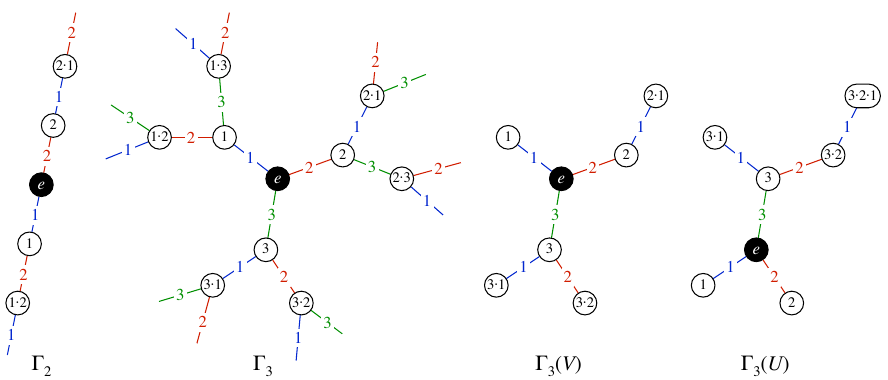} 
    \caption{In this example, $V = \{ e,\, 1,\, 2,\, 2{\cdot}1,\, 3,\, 3{\cdot}1,\, 3{\cdot}2 \} \subseteq G_3$ is a $3$-colour system and $U = "3V$. For example, $V[1] = U[1]$ and $V = V[2] \ne U[2] \ne U$.}\label{fig:cayley}
\end{figure}

In the reduced form, an element $x \in G_k$ is a product $x = c_1 c_2 \dotsm c_\ell$ such that $c_i \in [k]$ and $c_{i-1} \ne c_i$. The reduced form is unique; it corresponds to the sequence of edge colours along the unique path from $e$ to $x$ in $\Gamma_k$. We use the length of the path to define the norm \note{$|x|$\\$"x$\\$\tail$\\$\head$\\$\pred$}$|x| = \ell$.

We use the shorthand notation $"x = x^{-1}$ for the inverse of $x \in G_k$. If $x \in G_k-e$, there is a unique $c \in [k]$ such that $|xc| = |x|-1$; we say that $c$ is the \emph{tail} of $x$, in notation $\tail(x) = c$. We also define $\head(x) = \tail("x)$ and $\pred(x) = x \tail(x)$ for each $x \in G_k-e$.

We make the following observations: If $x,y \in G_k$, then $|"xy|$ is the length of the unique path from $x$ to $y$ in $\Gamma_k$; in particular, $d(x,y) = |"xy|$ defines a metric on $G_k$. If $|"xy| = 1$, nodes $x$ and $y$ are connected with an edge of colour $"xy$. We have $|"x| = |x|$ for all $x \in G_k$ and $|xy| \equiv |x| + |y| \mod 2$ for all $x,y \in G_k$. The equality $|xy| = |x| + |y|$ holds iff $x = e$, $y = e$, or $\tail(x) \ne \head(y)$.

If $V \subseteq G_k$ and $x \in G_k$, we define \note{$xV$\\$xf$}$xV = \{ xv : v \in V \}$. If $V \subseteq G_k$, $f\colon V \to X$, and $x \in G_k$, we also define the function $xf\colon xV \to X$ as follows: $(xf)(y) = f("xy)$ for each $y \in xV$. That is, $(xf)(xv) = f(v)$ for each $v \in V$.

\subsection{Colour Systems}

A non-empty set $V \subseteq G_k$ is a \emph{$k$-colour system} if $v \in V-e$ implies $\pred(v) \in V$. That is, a colour system is a prefix-closed subset; put otherwise, we can start from any $v \in V$ and walk towards $e$ in $\Gamma_k$ without leaving $V$. We define the set of edges \note{$E(V)$\\$\Gamma_k(V)$}$E(V) = \{ \{ \pred(v), v \} : v \in V-e \}$. Let $\Gamma_k(V)$ be the graph with the node set $V$ and the edge set $E(V)$. Now $\Gamma_k(V)$ is a connected subgraph of the tree $\Gamma_k$; see Figure~\ref{fig:cayley} for an example. Observe that if $\T$ is any $k$-edge-coloured tree, then we can construct a $k$-colour system $V \subseteq G_k$ such that $\T$ and $\Gamma_k(V)$ are isomorphic.

The following lemma is straightforward to verify.
\begin{lemma}
    If $V$ is a $k$-colour system and $u \in V$, then $"uV$ is a $k$-colour system. Moreover, $x \mapsto "ux$ is an isomorphism from $\Gamma_k(V)$ to $\Gamma_k("uV)$ that preserves adjacencies and edge colours.
\end{lemma}

For a colour system $V$ and integer $h$, we define \note{$V[h]$\\$f[h]$}$V[h] = \{ v\in V: |v| \le h \}$. Similarly, if $f\colon V \to X$, we define that $f[h]\colon V[h] \to X$ is the restriction of $f$ to $V[h]$. Note that $V[h]$ is a colour system. The set $u(("uV)[h]) \subseteq V$ consists of all nodes that are within distance $h$ from $u \in V$ in $\Gamma_k(V)$.

Let \note{$C(V,v)$\\$\deg$}$C(V,v) = \{ "uv : \{u,v\} \in E(V) \}$ denote the set of colours incident to $v \in V$ in $\Gamma_k(V)$. Note that $C(V,v) = \{ c \in [k] : vc \in V \} = ("vV)[1] - e$. The degree of $v$ is $\deg(V,v) = |C(V,v)|$, and colour system $V$ is said to be $d$-regular if $\deg(V,v) = d$ for all $v \in V$.

If $V$ is a colour system and $c \in C(V,e)$, we define \note{$\prune$}$\prune(V,c) = \{ v \in V-e : \head(v) \ne c \} + e$. Observe that $U = \prune(V,c)$ is a colour system. Moreover, if $V$ is $d$-regular, then $\deg(U,u) = d$ for all $u \in U-e$ and $\deg(U,e) = d-1$.

\subsection{Distributed Algorithms}

For the purposes of our lower-bound result, it is sufficient to define formally what a distributed algorithm $A$ outputs if we apply it in  $\Gamma_k(V)$, where $V$ is a colour system.

We already gave an informal definition of a distributed algorithm in Section~\ref{ssec:gmm}. In particular, we assumed that the nodes are anonymous (they do not have unique identifiers), and initially each node knows the colours of the incident edges. Put otherwise, initially a node $v \in V$ knows precisely $("vV)[1]$. Now if we let the nodes exchange all information that they have, after the first round each node $v \in V$ can reconstruct $("vV)[2]$, and recursively, after $r$ rounds each node knows precisely $("vV)[r+1]$. We will use this as our definition of a distributed algorithm.

Assume that $A$ is a function that associates a \emph{local output} $A(V,v)$ with any colour system $V$ and node $v \in V$. Then we say that $A$ is a \emph{distributed algorithm with running time $r$} if $("uU)[r+1] = ("vV)[r+1]$ implies $A(U,u) = A(V,v)$.

\subsection{Algorithms for Maximal Matchings}

\begin{wrapfigure}{r}{50mm}
    \raggedleft
    \vspace{-\intextsep}%
    \includegraphics[page=4]{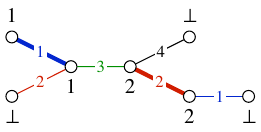}
    \vspace{-\intextsep}
\end{wrapfigure}
We say that a distributed algorithm $A$ \emph{finds a maximal matching} in colour system $V$ if
\begin{enumerate}[noitemsep,label=(M\arabic*)]
    \item\label{M1} we have $A(V,v) \in C(V,v) + \bot$ for each $v \in V$,
    \item\label{M2} if $A(V,v) = c \ne \bot$ then $vc \in V$ and $A(V,vc) = c$,
    \item\label{M3} if $A(V,v) = \bot$ and $c \in C(V,v)$ then $A(V,vc) \ne \bot$.
\end{enumerate}
The interpretation is that $A(V,v) = \bot$ if $v$ is unmatched and $A(V,v) = c \in C(v)$ if $v$ is matched along the edge of colour $c$. Property~\ref{M2} ensures that the output is consistent, and property~\ref{M3} ensures that the matching is maximal.

\section{Lower Bound}\label{sec:lb}

Let us first cover the case of $k \le 2$.

\begin{lemma}\label{lem:lbtriv}
    Let $k \le 2$ be a positive integer. A deterministic distributed algorithm that finds a maximal matching in any anonymous, $k$-edge-coloured graphs requires at least $k-1$ communication rounds.
\end{lemma}
\begin{proof}
    The case of $k = 1$ is trivial. Let us then focus on the case of $k = 2$. Define the $2$-colour systems $T = \{e,1\}$, $U = \{e,2\}$, and $V = \{e,1,2\}$. Now $A(T,1) = 1$ and $A(U,2) = 2$ for any distributed algorithm~$A$. However, we must have either $A(V,1) \ne 1$ or $A(V,2) \ne 2$, even though $("1T)[1] = ("1V)[1]$ and $("2U)[1] = ("2V)[1]$.
\end{proof}

The rest of this work contains the proof of the following theorem that covers the case of $k \ge 3$.
\begin{theorem}\label{thm:lb}
    Let \note{$k$\\$d$\\$A$}$k \ge 3$ be an integer, and let $d = k - 1$. Assume that $A$ is a distributed algorithm that finds a maximal matching in any $d$-regular $k$-colour system. Then there are two $d$-regular $k$-colour systems $U$ and $V$ such that $U[d] = V[d]$, $A(U,e) \ne \bot$, and $A(V,e) = \bot$.
\end{theorem}
\noindent In particular, the running time of $A$ is at least $d = k - 1$. Theorem~\ref{thm:lb2} follows.

\subsection{Overview of the Proof}

For the rest of this work, choose $k$, $d$, and $A$ as in the statement of Theorem~\ref{thm:lb}, and let $r$ be the running time of $A$. All colour systems are $k$-colour systems.

Sections \ref{ssec:templ}--\ref{ssec:crit} introduce a number of concepts that we will use to present our lower-bound construction. After that, we prove Theorem~\ref{thm:lb} by induction; the base case is in Section~\ref{ssec:base}, and the inductive step in Section~\ref{ssec:ind}.

\subsection{Templates and Colour Pickers}\label{ssec:templ}

An \emph{$h$-template} is a pair $(T,\tau)$ where $T \subseteq G_k$ is an $h$-regular colour system and $\tau \colon T \to [k]$ associates a \emph{forbidden colour} $\tau(t) \in [k] \setminus C(T,t)$ with each $t \in T$. The set of \emph{free colours} is \note{$F(T,\tau,t)$} $F(T,\tau,t) = [k] \setminus (C(T,t) + \tau(t))$ for each $t \in T$.

\begin{wrapfigure}{r}{75mm}
    \raggedleft
    \includegraphics[page=5]{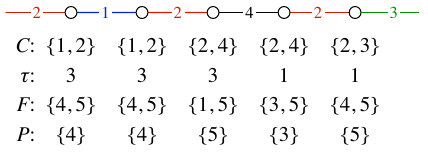} 
\end{wrapfigure}
Let $b$ be an integer with $0 \le b \le d-h$. A \emph{$b$-colour picker} for $(T,\tau)$ is a function $P$ that associates a subset $P(t) \subseteq F(T,\tau,t)$ of size $b$ with each node $t \in T$. That is, a $b$-colour picker chooses $b$ free colours for each node. The figure on the right gives an example with $h=2$, $b=1$, $d=4$, and $k=5$; a $2$-template is an infinite path and a $1$-colour picker chooses exactly one free colour for each node.

Let $P$ and $Q$ be colour pickers for $(T,\tau)$. We say that $P$ and $Q$ are \emph{disjoint} if $P(t) \cap Q(t) = \emptyset$ for all $t \in T$. If $P$ and $Q$ are disjoint colour pickers for $(T,\tau)$, we can construct a colour picker $R$ by setting $R(t) = P(t) \cup Q(t)$ for each $t \in T$.

\subsection{Extensions}

Let $(T,\tau)$ be an $h$-template and let $P$ be a $b$-colour picker for $(T,\tau)$. We will define a relation~$\leadsto$ between $G_k$ and $T$ recursively as follows; see Figure~\ref{fig:ext} for an illustration.
\begin{enumerate}[label=(\roman*)]
    \item We have $e \leadsto e$, $c \leadsto c$ for each $c \in C(T,e)$, and $c \leadsto e$ for each $c \in P(e)$.
    \item Assume that $x \leadsto t$ and $x \ne e$. \\
          We have $xc \leadsto tc$ for each $c \in C(T,t) - \tail(x)$,
          and $xc \leadsto t$ for each $c \in P(t) - \tail(x)$.
\end{enumerate}
We make the following observations.
\begin{enumerate}[noitemsep]
    \item If $x \leadsto t_1$ and $x \leadsto t_2$, we have $t_1 = t_2$.
    \item If $x \leadsto t$ and $x \ne e$, we have $\tail(x) \in C(T,t) \cup P(t)$.
    \item If $x \leadsto t$, $x \ne e$, and $\tail(x) \in C(T,t)$, we have $\pred(x) \leadsto t\tail(x)$.
    \item If $x \leadsto t$, $x \ne e$, and $\tail(x) \in P(t)$, we have $\pred(x) \leadsto t$.
    \item If $x \leadsto t$ and $c \in C(T,t)$, we have $xc \leadsto tc$.
    \item If $x \leadsto t$ and $c \in P(t)$, we have $xc \leadsto t$.
    \item If $x \leadsto t$ and $c \in [k] \setminus (C(T,t) \cup P(t))$, there is no $t' \in T$ with $xc \leadsto t'$.
    \item If $x \leadsto t$ then $|x| \ge |t|$.
    \item For each $t \in T$ there exists an $x$ such that $x \leadsto t$.
\end{enumerate}
\begin{figure}
    \centering
    \includegraphics[page=6]{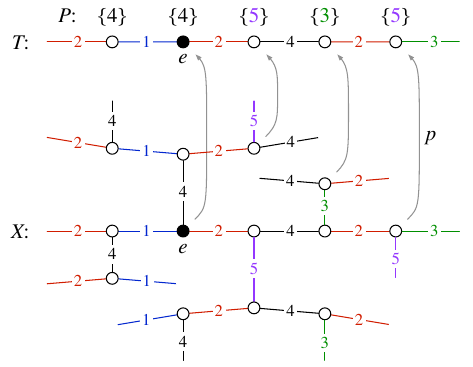} 
    \caption{Here $T$ is a $2$-template and $P$ is a $1$-colour picker. The arrows illustrate the relation~$\leadsto$ between $X$ and $T$, and hence also function $p$. In this case, $X$ is a $3$-regular colour system.}\label{fig:ext}
\end{figure}
Let $X = \{ x \in G_k : x \leadsto t \text{ for some } t \in T \}$. Define the function $p\colon X \to T$ as follows: for each $x \in X$, let $p(x)$ be the unique element with $x \leadsto p(x)$. Let $\xi = \tau \circ p$. We say that $(X,\xi,p)$ is the \emph{$P$-extension} of $(T,\tau)$, in notation, \note{$\ext$}$\ext(T,\tau,P) = (X,\xi,p)$.

\begin{remark}
    We can interpret extensions as universal covering graphs~\cite{angluin80local} as follows. First, consider the edge-coloured tree $\G = \Gamma_k(T)$. Then modify $\G$ as follows: for each $t \in T$ and $c \in P(t)$, add a self-loop of colour $c$ from $t$ to itself. Now $\G$ is an edge-coloured multigraph; then we construct the universal covering graph $\T$ of $\G$ (i.e., we ``unfold'' all self-loops of $\G$). Graph $\T$ is an edge-coloured tree; it can be verified that $\T$ is isomorphic to $\Gamma_k(X)$.
\end{remark}

\subsection{Properties of Extensions}

Let us first prove that an extension is a template.

\begin{lemma}\label{lem:ext}
    Assume that $(T,\tau)$ is an $h$-template, $P$ is a $b$-colour picker for $(T,\tau)$, and $(X,\xi,p) = \ext(T,\tau,P)$. Then $X$ is an $(h+b)$-regular colour system, and $(X,\xi)$ is an $(h+b)$-template. For each $x \in X$ we have $C(X,x) = C(T,p(x)) \cup P(p(x))$.
\end{lemma}
\begin{proof}
    Each $x \in X-e$ has $\pred(x) \in X$; hence $X$ is a colour system. If $x \in X$ and $c \in [k]$, we have $xc \in X$ iff $c \in C(T,p(x)) \cup P(p(x))$; hence $C(X,x) = C(T,p(x)) \cup P(p(x))$ and $\deg(x) = h+b$. It follows that $X$ is $(h+b)$-regular. By assumption, $\xi(x) = \tau(p(x)) \notin C(T,p(x)) \cup P(p(x)) = C(X,x)$; that is, $\xi$ associates a valid forbidden colour with each $x \in X$, and we conclude that $(X,\xi)$ is an $(h+b)$-template.
\end{proof}

Next, we observe that an extension has a high degree of symmetry.

\begin{lemma}\label{lem:extsym}
    Let $(X,\xi,p) = \ext(T,\tau,P)$, $x,y \in X$, and $p(x) = p(y)$. Then $"xX = "yX$, $"x\xi = "y\xi$, and $"xp = "yp$.
\end{lemma}

\begin{proof}
    Let $w \in "xX$. Assume that $w = c_1 c_2 \dotsm c_\ell$, where $c_i \in [k]$, and define $w_i = c_1 c_2 \dotsm c_i$. We have $w_i \in "xX$ and $xw_i \in x"xX = X$ for all $i$; let $t_i = p(xw_i)$.
    
    With these definitions, $xw_i \leadsto t_i$ for all $i = 0, 1, \dotsc, \ell$. We will prove by induction that $yw_i \leadsto t_i$ for all $i$. The base case of $i = 0$ is trivial. Now assume that $xw_i \leadsto t_i$ and $yw_i \leadsto t_i$. As we have $xw_ic_{i+1} \leadsto t_{i+1}$, there are two possibilities. If $c_{i+1} \in C(T,t_i)$, then $t_{i+1} = t_ic_{i+1}$ and $yw_ic_{i+1} \leadsto t_ic_{i+1}$. Otherwise $c_{i+1} \in P(t_i)$, $t_{i+1} = t_i$ and $yw_ic_{i+1} \leadsto t_i$. In both cases $yw_{i+1} \leadsto t_{i+1}$.
    
    It follows that $yw \leadsto t_\ell$, and hence $yw \in X$ with $p(yw) = t_\ell = p(xw)$. We have shown that $w \in "xX$ implies $w = "yyw \in "yX$ and $("yp)(w) = ("yp)("yyw) = p(yw) = p(xw) = ("xp)(w)$. By symmetry, $w \in "yY$ implies $w \in "xX$. Finally, $"xp = "yp$ implies $"x\xi = "y\xi$.
\end{proof}

We also show that the order in which we extend does not affect the end result. If we have two disjoint colour pickers $P$ and $Q$, we can first apply $P$ and then $Q$, or vice versa, and we obtain the same result as if we used the colour picker $t \mapsto P(t) \cup Q(t)$ directly; in this sense, extensions commute.
\begin{center}
    \includegraphics[page=7]{figs.pdf} 
\end{center}

\begin{lemma}\label{lem:extorder}
    Assume that $(T,\tau)$ is a template and $P$ and $Q$ are disjoint colour pickers for $(T,\tau)$. Let $R(t) = P(t) \cup Q(t)$ for each $t \in T$. Let $(K,\kappa,p) = \ext(T,\tau,P)$, $(L,\lambda,q) = \ext(K,\kappa,Q \circ p)$, and $(X,\xi,r) = \ext(T,\tau,R)$. Now $X = L$, $\lambda = \xi$, and $p \circ q = r$.
\end{lemma}

\begin{proof}
    Let $x = c_1 c_2 \dotsm c_\ell$, where $c_i \in [k]$, and define $x_i = c_1 c_2 \dotsm c_i$. We prove by induction that if $x_i \in X$, we also have $x_i \in L$ with $p(q(x_i)) = r(x_i)$, and if $x_i \notin X$, we also have $x_i \notin L$.
    
    The base case $i = 0$ is trivial: $p(q(e)) = p(e) = e = r(e)$ and $e \in X \cap L$. Now assume that $x_i \in X \cap L$ and $p(q(x_i)) = r(x_i)$. There are four cases depending on $c_{i+1}$:
    \begin{enumerate}
        \item Assume that $c_{i+1} \in C(T,r(x_i)) = C(T,p(q(x_i)))$. Then $c_{i+1} \in C(K,q(x_i))$, $x_{i+1} \in X \cap L$, and $p(q(x_{i+1})) = p(q(x_i c_{i+1})) = p(q(x_i) c_{i+1}) = p(q(x_i)) c_{i+1} = r(x_i) c_{i+1} = r(x_i c_{i+1}) = r(x_{i+1})$.
        \item Assume that $c_{i+1} \in P(r(x_i)) = P(p(q(x_i))) \subseteq R(r(x_i))$. Then $c_{i+1} \in C(K,q(x_i))$, $x_{i+1} \in X \cap L$, and $p(q(x_{i+1})) = p(q(x_i c_{i+1})) = p(q(x_i) c_{i+1}) = p(q(x_i)) = r(x_i) = r(x_i c_{i+1}) = r(x_{i+1})$.
        \item Assume that $c_{i+1} \in Q(r(x_i)) = Q(p(q(x_i))) \subseteq R(r(x_i))$. Then $c_{i+1} \in (Q \circ p)(q(x_i))$, $x_{i+1} \in X \cap L$, and $p(q(x_{i+1})) = p(q(x_i)) = r(x_i) = r(x_i c_{i+1}) = r(x_{i+1})$.
        \item Otherwise $x_{i+1} \notin X$ and $x_{i+1} \notin L$. As a consequence, $x_{i+j} \notin X$ and $x_{i+j} \notin L$ for all $j > 1$.
    \end{enumerate}
    In conclusion, we have $X = L$, $p \circ q = r$, and $\lambda = \tau \circ p \circ q = \tau \circ r = \xi$.
\end{proof}

\subsection{Realisations}

Let $(T,\tau)$ be an $h$-template. Define a $(d-h)$-colour picker $P$ by setting $P(t) = F(T,\tau,t)$ for each $t \in T$. Let $(V,g,p) = \ext(T,\tau,P)$. We say that $(V,p)$ is the \emph{realisation} of template $(T,\tau)$, in notation, \note{$\real$}$(V,p) = \real(T,\tau)$.

Intuitively, $V$ is a concrete problem instance---it is always a $d$-regular colour system, and hence we can apply algorithm $A$ to $V$. Templates can be seen as compact, schematic representations of problem instances.

Lemma~\ref{lem:extsym} has the following corollary.

\begin{corollary}\label{cor:realequiv}
    Let $(V,p) = \real(T,\tau)$. If $u, v \in V$ and $p(u) = p(v)$, then $"uV = "vV$. In particular, $A(V,u) = A(V,v)$.
\end{corollary}

Put otherwise, if $(T,\tau)$ is a template with the realisation $(V,p)$, each node $t \in T$ represents an \emph{equivalence class} $p^{-1}(t) \subseteq V$ of nodes with identical outputs. For each $t \in T$, we define \note{$A(T,\tau,t)$}$A(T,\tau,t) = A(V,v)$ where $v \in p^{-1}(t)$; by Corollary~\ref{cor:realequiv}, this does not depend on the choice of $v$.

We define \note{$M(T,\tau)$}$M(T,\tau) = \{ \{u,v\} \in E(T) : A(T,\tau,u) = A(T,\tau,v) = "uv \}$. Note that $M(T,\tau)$ is always a matching in the tree $\Gamma_k(T)$, but the matching is not necessarily maximal. If $S \subseteq T$, we also define \note{$M(\cdot,\cdot,\cdot)$}$M(T,S,\tau) = \{ \{u,v\} \in M(T,\tau) : u, v \in S \}$, the restriction of $M(T,\tau)$ to $S$.

Lemma~\ref{lem:extorder} has the following corollary; it shows that a template and its extensions have the same realisations.

\begin{corollary}\label{cor:extreal}
    Let $(K,\kappa,p) = \ext(T,\tau,P)$, $(X,r) = \real(T,\tau)$, and $(L,q) = \real(K,\kappa)$. Then $X = L$, $p \circ q = r$, and $A(K,\kappa,x) = A(T,\tau,p(x))$ for all $x \in K$.
\end{corollary}

The following lemma is yet another application of the symmetry that we have in extensions: if a template has free colours (i.e., $h < d$), then an algorithm produces a perfect matching in the realisation of the template.

\begin{lemma}\label{lem:realperf}
    Assume that $(T,\tau)$ is an $h$-template with $h < d$. Then $A(T, \tau, t) \ne \bot$ for all $t \in T$.
\end{lemma}
\begin{proof}
    Let $(V,p) = \real(T,\tau)$, $t \in T$, and $v \in p^{-1}(t)$. If $h < d$, there exists a $c \in F(T,\tau,t)$. Let $u = vc$; we have $p(u) = p(v) = t$, $c \in C(V,v)$, and $A(V,u) = A(V,v) = A(T, \tau, t)$. Now $A(T, \tau, t) = \bot$ would contradict property~\ref{M3}.
\end{proof}

\subsection{Zero-Templates}

Let \note{$Z$\\$\hat c$}$Z = \{ e \}$ be the colour system with only one node. For each $c \in [k]$, let $\hat{c}$ denote the function $\hat{c}\colon Z \to [k]$ that maps $\hat{c}(e) = c$. Now $(Z,\hat{c})$ is a $0$-template for each $c \in [k]$.

If $A$ is the greedy algorithm, we have $A(Z, \hat 1, e) = 2$ and $A(Z, \hat 3, e) \ne 2$. The following lemma generalises this observation.

\begin{lemma}\label{lem:0temp}
    There are distinct colours $c_1, c_2, c_3 \in [k]$ such that
    $A(Z, \hat c_1, e) = c_2$ and
    $A(Z, \hat c_3, e) \ne c_2$.
\end{lemma}

\begin{proof}
    For each $c \in [k]$, let $h(c) = A(Z, \hat c, e)$. By Lemma~\ref{lem:realperf}, we have $h(c) \in [k]$ for each $c \in [k]$. Moreover, $h(c) \in [k] - \hat{c}(e) = [k] - c$. Hence we have a function $h\colon [k] \to [k]$ that does not have any fixed points.
    \begin{enumerate}
        \item Assume that $h(h(1)) \ne 1$. Then we can choose $c_1 = h(1)$, $c_2 = h(h(1))$, and $c_3 = 1$.
        \item Assume that $h(h(1)) = 1$. Let $c \in [k] - \{ 1, h(1) \}$. If $h(c) = h(1)$, we can choose $c_1 = h(1)$, $c_2 = 1$, and $c_3 = c$. If $h(c) \ne h(1)$, we can choose $c_1 = 1$, $c_2 = h(1)$, and $c_3 = c$. \qedhere
    \end{enumerate}
\end{proof}

\subsection{Compatible Templates and Critical Pairs}\label{ssec:crit}

Let $h \ge 1$. We say that templates $(S,\sigma)$ and $(T,\tau)$ are \emph{$h$-compatible} if
\begin{enumerate}[label=(C\arabic*),noitemsep]
    \item\label{C1} $S[h] = T[h]$,
    \item\label{C2} $\sigma[h-1] = \tau[h-1]$.
\end{enumerate}
We emphasise that $h$-compatible templates are not necessarily $h$-templates.

We say that $(S,\sigma)$ and $(T,\tau)$ form an \emph{$h$-critical pair} if they are $h$-compatible $h$-templates and they satisfy the following additional properties:
\begin{enumerate}[resume*]
    \item\label{C3} $A(T,\tau,e) \notin C(T,e)$,
    \item\label{C4} $A(S,\sigma,s) \in C(S,s)$ for each $s \in S$.
\end{enumerate}
If $h < d$, Lemma~\ref{lem:realperf} and property~\ref{C3} imply that $A(T,\tau,e) \in F(T,\tau,e)$. Property~\ref{C4} implies that $M(S,\sigma)$ is a perfect matching in $\Gamma_k(S)$, while property~\ref{C3} implies that $M(T,\tau)$ cannot be a perfect matching in $\Gamma_k(T)$.

\subsection{Base Case}\label{ssec:base}

In this section we show that there exists a $1$-critical pair. Choose $c_1, c_2, c_3 \in [k]$ as in Lemma~\ref{lem:0temp} and let $c_4 = A(Z, \hat c_3, e)$. Note that $c_4 \ne c_2$; however, we may have $c_4 = c_1$.

Let $K = L = X = \{ e, c_2 \}$. Define $\kappa(e) = \kappa(c_2) = \xi(e) = c_1$ and $\lambda(e) = \lambda(c_2) = \xi(c_2) = c_3$. Now $(K,\kappa)$, $(L,\lambda)$, and $(X,\xi)$ are $1$-templates; the construction is illustrated below:
\begin{center}
    \includegraphics[page=8]{figs.pdf} 
\end{center}
If $p(e) = p(c_2) = e$ and $P(e) = c_2$, we have $(K,\kappa,p) = \ext(Z,\hat c_1,P)$ and $(L,\lambda,p) = \ext(Z,\hat c_3,P)$. Therefore $A(K,\kappa,v) = c_2$ for each $v \in K$ and $A(L,\lambda,v) = c_4$ for each $v \in L$.

Now we construct $1$-templates \note{$S_1,\sigma_1$\\$T_1,\tau_1$}$(S_1,\sigma_1)$ and $(T_1,\tau_1)$ as follows:
\begin{enumerate}[label=(\roman*),noitemsep]
    \item\label{basecase1} If $A(X,\xi,e) \ne c_2$, we choose $(S_1,\sigma_1) = (K,\kappa)$ and $(T_1,\tau_1) = (X,\xi)$.
    \item\label{basecase2} If $A(X,\xi,e) = c_2$, we choose $(S_1,\sigma_1) = ("c_2 X,"c_2 \xi)$ and $(T_1,\tau_1) = ("c_2 L,"c_2 \lambda)$.
\end{enumerate}

\begin{lemma}\label{lem:base}
    Templates $(S_1,\sigma_1)$ and $(T_1,\tau_1)$ form a $1$-critical pair.
\end{lemma}
\begin{proof}
    We have $S_1[1] = T_1[1] = K = L = X = \{ e, c_2 \}$, verifying property~\ref{C1}. To verify \ref{C2}, note that case~\ref{basecase1} implies $\sigma_1(e) = \tau_1(e) = c_1$ and case~\ref{basecase2} implies $\sigma_1(e) = \tau_1(e) = c_3$. To verify property~\ref{C3}, observe that $A(T_1,\tau_1,e) \ne c_2$ while $C(T_1,e) = \{ c_2 \}$. To verify property~\ref{C4}, observe that $A(S_1,\sigma_1,s) = c_2$ and $C(S_1,s) = \{ c_2 \}$ for all $s \in S_1$.
\end{proof}

\subsection{Inductive Step}\label{ssec:ind}

Now assume that $(S_h,\sigma_h)$ and $(T_h,\tau_h)$ form an $h$-critical pair, where $1 \le h < d$. In this section, we will construct an $(h+1)$-critical pair $(S_{h+1},\sigma_{h+1})$ and $(T_{h+1},\tau_{h+1})$.

Recall that Lemma~\ref{lem:realperf} implies that $A(S_h,\sigma_h,s) \ne \bot$ for all $s \in S_h$ and $A(T_h,\tau_h,t) \ne \bot$ for all $t \in T_h$. We define two colour pickers as follows; see Figures \ref{fig:ind} and \ref{fig:ind2} for illustrations.
\begin{figure}
    \begin{center}
        \includegraphics[page=9]{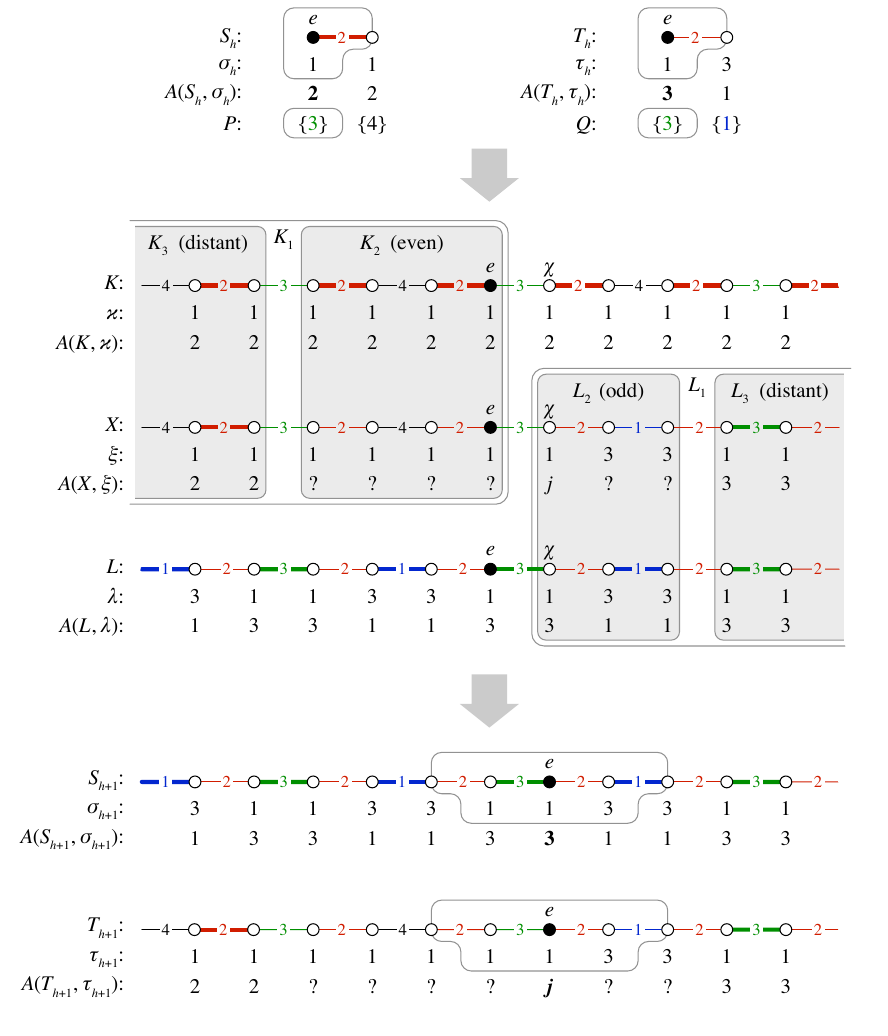} 
    \end{center}
    \caption{Inductive step. In this example, $h = 1$ and $\chi = 3$. We assume that $j \notin \{2,3\}$, and thus we can choose $y = \chi$ in Lemma~\ref{lem:critnode}.}\label{fig:ind}
\end{figure}
\begin{enumerate}[label=(\roman*)]
    \item Define a $1$-colour picker $Q$ for $(T_h,\tau_h)$ as follows. Let $t \in T_h$. If $A(T_h,\tau_h,t) \in F(T_h,\tau_h,t)$, we choose $Q(t) = \{ A(T_h,\tau_h,t) \}$. Otherwise we choose an arbitrary free colour $c \in F(T_h,\tau_h,t)$ and set $Q(t) = \{ c \}$.
    \item Define a $1$-colour picker $P$ for $(S_h,\sigma_h)$ as follows. Let $s \in S_h$. If $|s| \le h-1$, we have $s \in T_h$ and $F(S_h,\sigma_h,s) = F(T_h,\tau_h,s)$; hence we can choose $P(s) = Q(s)$. Otherwise we choose an arbitrary free colour $c \in F(S_h,\sigma_h,s)$ and set $P(s) = \{ c \}$.
\end{enumerate}
Let $(K,\kappa,p) = \ext(S_h,\sigma_h,P)$, $(L,\lambda,q) = \ext(T_h,\tau_h,Q)$, and $\chi = A(T_h,\tau_h,e)$. We make the following observations:
\begin{enumerate}[noitemsep]
    \item $(K,\kappa)$ and $(L,\lambda)$ are $(h+1)$-templates,
    \item $(K,\kappa)$ and $(L,\lambda)$ are $h$-compatible,
    \item $\{e,\chi\} \in E(K)$ and $\{e,\chi\} \in E(L)$,
    \item $p(e) = p(\chi) = e$ and $q(e) = q(\chi) = e$,
    \item $"\chi K = K$, $"\chi \kappa = \kappa$, $"\chi L = L$, and $"\chi \lambda = \lambda$,
    \item $A(K,\kappa,v) \in C(K,v)$ for each $v \in K$, i.e., $M(K,\kappa)$ is a perfect matching in $\Gamma_k(K)$,
    \item $A(L,\lambda,v) \in C(L,v)$ for each $v \in L$, i.e., $M(L,\lambda)$ is a perfect matching in $\Gamma_k(L)$,
    \item $\{e,\chi\} \notin M(K,\kappa)$ but $\{e,\chi\} \in M(L,\lambda)$.
\end{enumerate}

\begin{figure}
    \begin{center}
        \includegraphics[page=10]{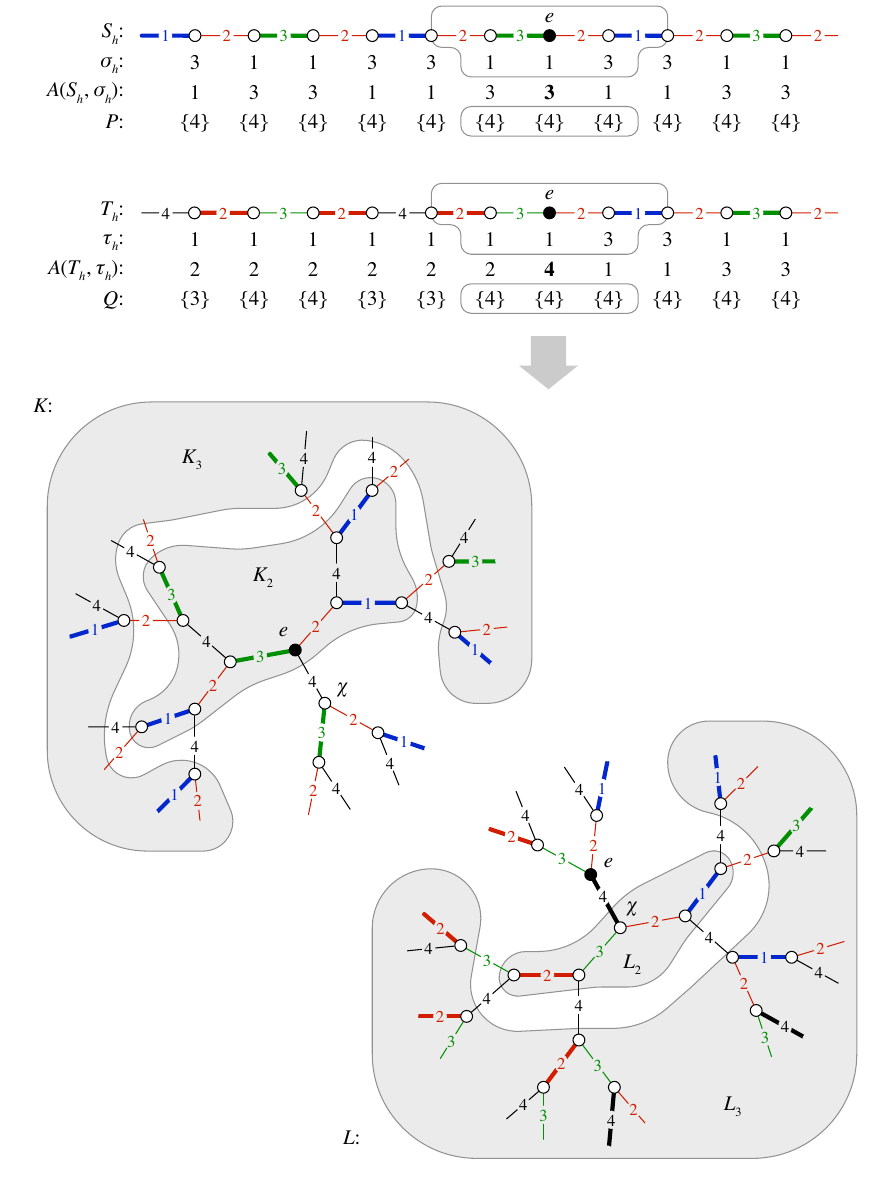} 
    \end{center}
    \caption{Inductive step. In this example, $h = 2$ and $\chi = 4$.}\label{fig:ind2}
\end{figure}

Now we will use $(K,\kappa)$ and $(L,\lambda)$ to construct a new $(h+1)$-template $(X,\xi)$; refer to Figure~\ref{fig:ind}. Let
$K_1 = \prune(K,\chi)$,
$L_1 = \chi \prune("\chi L,\chi)$, and
$X = K_1 \cup L_1$.
Define $\xi\colon X \to [k]$ as follows: $\xi(v) = \kappa(v)$ for all $v \in K_1$ and $\xi(v) = \lambda(v)$ for all $v \in L_1$. We make the following observations:
\begin{enumerate}[noitemsep]
    \item $(X,\xi)$ is an $(h+1)$-template,
    \item $(X,\xi)$, $(K,\kappa)$, and $(L,\lambda)$ are pairwise $h$-compatible,
    \item $("\chi X,"\chi \xi)$, $("\chi K,"\chi \kappa)$, and $("\chi L,"\chi \lambda)$ are pairwise $h$-compatible.
    \item $("yX,"y\xi)$ and $("yK,"y\kappa)$ are $(h+1)$-compatible for any $y \in K_1$,
    \item $("yX,"y\xi)$ and $("yL,"y\lambda)$ are $(h+1)$-compatible for any $y \in L_1$.
\end{enumerate}
Hence we have a family of $(h+1)$-compatible $(h+1)$-templates; however, we need to construct an $(h+1)$-critical pair.

\begin{lemma}\label{lem:critnode}
    There exists a node $y \in X$ such that $A(X,\xi,y) \notin C(X,y)$.
\end{lemma}
\begin{proof}
    We say that an edge $\{u,v\}$ is \emph{distant} if $|u| > r+1$ and $|v| > r+1$; otherwise it is \emph{near}.

    Set $M(K,\kappa)$ is a perfect matching in $\Gamma_k(K)$. Moreover, $\{e,\chi\} \notin M(K,\kappa)$; therefore we have either $\{u,v\} \subseteq K_1$ or $\{u,v\} \cap K_1 = \emptyset$ for each $\{u,v\} \in M(K,\kappa)$. It follows that $\bigcup M(K,K_1,\kappa) = K_1$. Let $K'_3 \subseteq M(K,K_1,\kappa)$ consists of the edges that are distant, and let $K'_2 = M(K,K_1,\kappa) \setminus K'_3$ consist of the edges that are near. Define $K_2 = \bigcup K'_2$ and $K_3 = \bigcup K'_3$; see Figure~\ref{fig:ind} for an illustration.

    Set $M(L,\lambda)$ is a perfect matching in $\Gamma_k(L)$. Moreover, $\{e,\chi\} \in M(L,\lambda)$; this is the unique edge that joins $L_1$ and $L \setminus L_1$. Therefore we have $\bigcup M(L,L_1,\lambda) = L_1 - \chi$. Let $L'_3 \subseteq M(L,L_1,\lambda)$ consists of the edges that are distant, and let $L'_2 = M(L,L_1,\lambda) \setminus L'_3$ consist of the edges that are near. Define $L_2 = (\bigcup L'_2) + \chi$ and $L_3 = \bigcup L'_3$.
    
    It follows that
    \begin{enumerate}[noitemsep]
        \item $K_3$, $K_2$, $L_2$, and $L_3$ form a partition of $X$,
        \item $("vK)[r+1] = ("vX)[r+1]$ and $("v\kappa)[r+1] = ("v\xi)[r+1]$ for any $v \in K_3$,
        \item $("vL)[r+1] = ("vX)[r+1]$ and $("v\lambda)[r+1] = ("v\xi)[r+1]$ for any $v \in L_3$,
        \item $A(K,\kappa,v) = A(X,\xi,v)$ for any $v \in K_3$,
        \item $A(L,\lambda,v) = A(X,\xi,v)$ for any $v \in L_3$,
        \item $\{u,v\} \in K'_3 \cup L'_3$ implies $\{u,v\} \in M(X,\xi)$,
        \item $K_2$ is a finite set with an even number of nodes,
        \item $L_2$ is a finite set with an odd number of nodes.
    \end{enumerate}
    By a parity argument, there is a node $y \in K_2 \cup L_2$ such that $y \notin \bigcup M(X,\xi)$, i.e., $A(X,\xi,y) \notin C(X,y)$.
\end{proof}

Now choose $y$ as in Lemma~\ref{lem:critnode}, and define $(S_{h+1},\sigma_{h+1})$ and $(T_{h+1},\tau_{h+1})$ as follows:
\begin{enumerate}[noitemsep]
    \item If $y \in K_1$, define\note{$S_i,\sigma_i$\\$T_i,\tau_i$}
        $S_{h+1} = "yK$,
        $\sigma_{h+1} = "y\kappa$,
        $T_{h+1} = "yX$, and
        $\tau_{h+1} = "y\xi$.
    \item If $y \in L_1$, define
        $S_{h+1} = "yL$,
        $\sigma_{h+1} = "y\lambda$,
        $T_{h+1} = "yX$, and
        $\tau_{h+1} = "y\xi$.
\end{enumerate}

\begin{lemma}\label{lem:ind}
    Templates $(S_{h+1},\sigma_{h+1})$ and $(T_{h+1},\tau_{h+1})$ form an $(h+1)$-critical pair.
\end{lemma}
\begin{proof}
    First, assume that $y \in K_1$. We have already observed that $(S_{h+1},\sigma_{h+1}) = ("yK,"y\kappa)$ and $(T_{h+1},\tau_{h+1}) = ("yX,"y\xi)$ are $(h+1)$-compatible. Moreover, we have $A(T_{h+1},\tau_{h+1},e) = A(X,\xi,y) \notin C(X,y) = C(T_{h+1},e)$ and $A(S_{h+1},\sigma_{h+1},s) = A(K,\kappa,ys) \in C(K,ys) = C(S_{h+1},s)$ for each $s \in S_{h+1}$. Hence $(S_{h+1},\sigma_{h+1})$ and $(T_{h+1},\tau_{h+1})$ form an $(h+1)$-critical pair.
    
    The case of $y \in L_1$ is analogous.
\end{proof}

By induction, there are $d$-templates $(S_d,\sigma_d)$ and $(T_d,\tau_d)$ that form a $d$-critical pair. Theorem~\ref{thm:lb} follows by choosing $U = S_d$ and $V = T_d$.

\section*{Acknowledgements}

We thank Mika G\"o\"os for comments and feedback, and Petteri Kaski, Christoph Lenzen, Joel Rybicki, and Roger Wattenhofer for discussions. This work was supported in part by the Academy of Finland, Grants 132380 and 252018, the Research Funds of the University of Helsinki, and the Finnish Cultural Foundation.

\end{document}